\documentclass{article}

\usepackage{amsmath, amsthm, amssymb}
\usepackage{enumerate}
\usepackage{tabularx}
\usepackage{array}
\usepackage{fullpage}
\usepackage{graphicx}

\newtheorem{theorem}{Theorem}

\newtheorem{corollary}{Corollary}

\theoremstyle{definition}
\newtheorem{remark}{Remark}

\newcommand{\B}{\mathcal{B}}
\DeclareMathOperator{\dist}{\mathrm{dist}}
\DeclareMathOperator{\E}{\mathrm{E}}

\usepackage[numbers]{natbib}
\title{Thorup-Zwick Emulators are Universally Optimal Hopsets\thanks{Supported by NSF Grants CCF-1514383 and CCF-1637546.}}
\author{Shang-En Huang\\ University~of Michigan\and Seth Pettie\\ University~of Michigan}

\begin{document}
\maketitle

\begin{abstract}
A $(\beta,\epsilon)$-\emph{hopset} is, informally, a weighted edge set that, when added to a graph, allows one to get from point $a$ to point $b$ using a path with at most $\beta$ edges (``hops'') and length $(1+\epsilon)\dist(a,b)$.  In this paper we observe that Thorup and Zwick's \emph{sublinear additive} emulators are also actually $(O(k/\epsilon)^k,\epsilon)$-hopsets for every $\epsilon>0$, and that with a small change to the Thorup-Zwick construction, the size of the hopset can be made $O(n^{1+\frac{1}{2^{k+1}-1}})$.
As corollaries, we also shave 
``$k$'' factors off the size of
Thorup and Zwick's~\cite{TZ06} sublinear additive emulators 
and the sparsest known
$(1+\epsilon,O(k/\epsilon)^{k-1})$-spanners, 
due to Abboud, Bodwin, and Pettie~\cite{ABP17}.
\end{abstract}

\section{Introduction}

Let $G=(V, E, w)$ be a weighted undirected graph. Define $\dist_G^{(\beta)}(u, v)$ to be the length of the shortest path from $u$ to $v$ in $G$ \emph{that uses at most $\beta$ edges}, or ``hops.'' 
Whereas 
$\dist_G = \dist_G^{(\infty)}$ 
is a metric, 
$\dist_G^{(\beta)}$ is not in general.
A set $H\subset {V\choose 2}$ of weighted edges is called a \emph{$(\beta, \epsilon)$-hopset} if for every $u, v\in V$, $$\dist_G(u, v)\le \dist_{G\cup H}^{(\beta)}(u, v)\le (1+\epsilon)\dist_G(u, v).$$

\paragraph{Background.}
Cohen~\cite{Coh00} formally defined the notion of a hopset, but the idea was latent in earlier work~\cite{UY91,KS97,Cohen97,SS99}.  Cohen's $(\beta,\epsilon)$-hopset had size $O(n^{1+1/\kappa}\log n)$ and 
$\beta = (\epsilon^{-1}\log n)^{O(\log\kappa)}$.  Elkin and Neiman~\cite{ElkinNeiman16} showed that a constant hopbound $\beta$ suffices (when $\kappa,\epsilon$ are constants).
In particular, their hopset has size $O(n^{1+1/\kappa}\log n\log \kappa)$ and
$\beta = O(\epsilon^{-1}\log\kappa)^{\log\kappa}$.  Abboud, Bodwin, and Pettie~\cite{ABP17} recently proved that the tradeoffs of~\cite{ElkinNeiman16} are essentially optimal: for any integer $k$, any hopset of size $n^{1+\frac{1}{2^{k+1}-1}-\delta}$ must have $\beta = \Omega(c_k / \epsilon^{k+1})$, where $c_k$ is a constant depending only on $k$.\footnote{Note that setting $\kappa = 2^{k+1}-1$ in the Elkin-Neiman construction gives $\beta = O(k/\epsilon)^k$, where $\log\kappa = \lfloor\log\kappa\rfloor = k$.  Thus, saving any $\delta$ in the exponent of the hopset increases $\beta$ significantly.  In general, the statement of~\cite{ElkinNeiman16} obscures the nature of the tradeoff: there are \emph{not} distinct tradeoffs for each $\kappa \in \{1,2,3,\ldots\}$,
but only for $\kappa \in \{1,3,7,\ldots,2^{k+1}-1,\ldots\}$.}
There are other constructions of hopsets~\cite{Bernstein09, HKN14, HKN16, MPVX15} that are designed for parallel or dynamic environments; their tradeoffs (between hopset size and hopbound) are worse than~\cite{Coh00,ElkinNeiman16} and the ones presented here.  
See Table~\ref{table:previous-works}.

\renewcommand{\arraystretch}{1.4}
\begin{table}[h]
\centering
\begin{tabular}{|l|l|l|l|}
\hline
{\textbf{Authors}} & {\textbf{Size}} & \textbf{Hopbound} & \textbf{Stretch}\\
\hline
Klein and Subramanian~\cite{KS97} & $O(n)$    & $O(\sqrt{n}\log n)$  & 1\\\hline
Thorup and Zwick~\cite{TZ05} & $O(\kappa n^{1+1/\kappa})$ & 2 & $2\kappa-1$\\\hline
Cohen  \cite{Coh00} & $O(n^{1+\frac{1}{\kappa}}\cdot \log n)$ & $((\log n)/\epsilon)^{O(\log \kappa)}$ & $1+\epsilon$ \\
\hline
Elkin and Neiman \cite{ElkinNeiman16} & $O(n^{1+\frac{1}{\kappa}} \log n\log\kappa)$ & $O((\log \kappa)/\epsilon)^{\log \kappa}$ & $1+\epsilon$\\
\hline
Abboud, Bodwin, and Pettie \cite{ABP17} & 
\multicolumn{1}{l}{$n^{1+\frac{1}{2^{k+1}-1}-\delta}$} $\longrightarrow$
& 
$\Omega(c_k / \epsilon^{k+1})$ & $1+\epsilon$\\
\hline\hline
{\textbf{New}} & $O\left(n^{1+\frac{1}{2^{k+1}-1}}\right)$ & $O(k/\epsilon)^k$ & $1+\epsilon$\\\hline

\end{tabular}
\caption{Tradeoffs between size and hopbound of previous hopsets.  Fix the parameter $\kappa=2^{k+1}-1$ to compare~\cite{Coh00,ElkinNeiman16} against the lower bound~\cite{ABP17} and the new result.}\label{table:previous-works}
\end{table}

\paragraph{Hopsets, Emulators, and Spanners.} Recall that $G$ is an undirected graph, possibly weighted.  A \emph{spanner} is a subgraph of $G$ such that $\dist_H(u,v) \le f(\dist_G(u,v))$ 
for some nondecreasing \emph{stretch function $f$}.
An \emph{emulator} of an \emph{unweighted} graph $G$ is a \emph{weighted} edge set $H$ such that $\dist_H(u,v) \in [\dist_G(u,v), f(\dist_G(u,v))]$.  Syntactically, the definition of hopsets is closely related to emulators.  The difference is that hopsets have a hopbound constraint but are allowed to use original edges in $G$ whereas emulators must use only $H$.  The purpose of emulators is to \emph{compress} the graph metric $\dist_G$: 
ideally $|H| \ll |E(G)|$.  Historically, the literature on hopset constructions~\cite{Coh00,ElkinNeiman16} has been noticeably more complex than those of spanners and emulators, many of which~\cite{Althofer+93,ACIM99,DHZ00,TZ06,BKMP10,Knudsen14,ABP17} are quite elegant.  Our goal in this work is to demonstrate that there is nothing \emph{intrinsically} complex about hopsets, and that a very simple construction improves on all prior constructions and matches the Abboud-Bodwin-Pettie lower bound.

\paragraph{New Results.} Thorup and Zwick~\cite{TZ06} designed their emulator for \emph{unweighted} graphs, and proved that it has size $O(kn^{1+\frac{1}{2^{k+1}-1}})$ and a sublinear additive stretch function $f(d) = d + O(kd^{1-1/k})$.  In this paper we show that the Thorup-Zwick emulator, when applied to a \emph{weighted} graph, produces a $(\beta,\epsilon)$-hopset that achieves every point on the Abboud-Bodwin-Pettie~\cite{ABP17} lower bound tradeoff curve.  Moreover, with two subtle modifications to the construction, we can reduce the size to $O(n^{1+\frac{1}{2^{k+1}-1}})$, shaving off a factor $k$.  Our technique also applies to other constructions, and as corollaries we improve the size of Thorup and Zwick's emulator~\cite{TZ06} and Abboud, Bodwin, and Pettie's $(1+\epsilon,\beta)$-spanners.\footnote{A $(1+\epsilon,\beta)$-spanner of an unweighted graph is one with stretch function $f(d) = (1+\epsilon)d + \beta$.}

\begin{theorem}\label{thm:hopset}
Fix any weighted graph $G$ and integer $k\ge 1$. 
There is a $(\beta,\epsilon)$-hopset for $G$ with size $O(n^{1+\frac{1}{2^{k+1}-1}})$ and $\beta = 2\left(\frac{(4+o(1))k}{\epsilon}\right)^k$. 
\end{theorem}

\begin{theorem}\label{thm:emulator} (cf.~\cite{TZ06})
Fix any unweighted graph $G$ and integer $k\ge 1$.
There is a sublinear additive emulator $H$ for $G$ with 
size $O(n^{1+\frac{1}{2^{k+1}-1}})$ and stretch function 
$f(d) = d + (4+o(1))kd^{1-1/k}$.
\end{theorem}

\begin{theorem}\label{thm:spanner} (cf.~\cite{ABP17})
Fix any unweighted graph $G$, integer $k\ge 1$, and real $\epsilon >0$.  There is a $(1+\epsilon, ((4+o(1))k/\epsilon)^{k-1})$-spanner $H$ for $G$ with size $O((k/\epsilon)^h n^{1+\frac{1}{2^{k+1}-1}})$, 
where $h = \frac{3\cdot 2^{k-1} - (k+2)}{2^{k+1}-1} < 3/4$.
\end{theorem}

\begin{remark}
In recent and independent technical report, 
Elkin and Neiman~\cite{ElkinNeiman17} 
also observed that Thorup and Zwick's emulator 
yields an essentially optimal hopset.  
They proposed a modification to Thorup and Zwick's construction that reduces the size to $O(n^{1+\frac{1}{2^{k+1}-1}})$ (eliminating a factor $k$), 
but \emph{increases} the hopbound $\beta$ from $O(k/\epsilon)^k$ to $O((k+1)/\epsilon)^{k+1}$.  For example, their technique does not imply any of 
the improvements found in Theorems~\ref{thm:hopset}, \ref{thm:emulator}, or \ref{thm:spanner}.
\end{remark}

\section{The Hopset Construction}

In this section, we present the construction of the hopset 
based on Thorup and Zwick's emulator~\cite{TZ06},
then analyze its size, stretch, and hopbound.

The construction is parameterized by an integer $k\ge 1$
and a set $\{q_i\}$ of sampling probabilities.
Let $V=V_0\supseteq V_1 \supseteq V_2 \supseteq \cdots \supseteq V_k\supseteq V_{k+1}=\emptyset$ be the vertex sets in each layer.
For each $i\in [0,k)$, each vertex in $V_i$ is independently promoted 
to $V_{i+1}$ with probability $q_{i+1}/q_i$.
Thus $\E[|V_i|] = nq_i$.
For each vertex $v\in V$ and $i\in [1,k]$, define $p_i(v)$ to 
be any vertex in $V_i$ such that $\dist_G(v, p_i(v)) = \dist_G(v, V_i)$. 
For any vertex $v\in V_i\setminus V_{i+1}$, define $\B(v)$ to be:
\[
\B(v) = \{ u\in V_i\ |\ \dist_G(v, u) < \dist_G(v, p_{i+1}(v))\}
\]
Note that $p_{k+1}(v)$ does not exist; by convention $\dist_G(v, p_{k+1}(v))=\infty$.
The hopset is defined to be $H= E_0\cup E_1\cup\cdots\cup E_k$, where
\[E_i = \bigcup_{v\in V_i\setminus V_{i+1}}\left\{(v, u)\ |\ u\in \mathcal{B}(v) \cup \{ p_{i+1}(v)\}\right\}.\] 
The length of an edge in $H$ is always the distance between its endpoints.
This concludes the description of the construction.

\subsection{Size Analysis}

The expected size of 
$E_i$ is at most $\E[|V_i|](q_i/q_{i+1})=nq_i^2/q_{i+1}$, 
for each $i\in[0,k)$, and is $(nq_k)^2$ if $i=k$. 
Following Pettie~\cite{P09}, 
we choose $\{q_i\}$ such that the 
layers of the hopset have geometrically decaying sizes.
Setting $q_i = n^{-\frac{2^i-1}{2^{k+1}-1}} 
\cdot 2^{-2^i-i+1}$, the expected size of $E_i$, for $i\in[0,k)$, is 
\begin{align*}
    nq_i^2/q_{i+1} &=
    n \cdot \left(n^{-\frac{2^i-1}{2^{k+1}-1}} \cdot
    2^{-2^{i}-i+1} \right)^2
    \bigg/ \left(n^{-\frac{2^{i+1}-1}{2^{k+1}-1}} \cdot
    2^{-2^{i+1}-i} \right)
     \\
    &= n^{1 - \frac{2^{i+1}-2}{2^{k+1}-1} + \frac{2^{i+1}-1}{2^{k+1}-1}}\cdot 2^{-2^{i+1}-2i+2-(-2^{i+1}-i)}\\
    &= n^{1 + \frac{1}{2^{k+1}-1}}\cdot
    2^{-i+2}.
\end{align*}
The expected size of $E_k$ is
\begin{align*}
    (nq_k)^2 &= n^2\cdot \left(n^{-\frac{2^k-1}{2^{k+1}-1}}\cdot 2^{-2^k-k+1}\right)^2\\
    &= n^{1 + \frac{1}{2^{k+1}-1}} 2^{-2^{k+1}-2k+2} \;\ll\; n^{1 + \frac{1}{2^{k+1}-1}} \cdot 2^{-k+2},
\end{align*}
so the expected size of $H$ is at most
$$
\sum_{i=0}^k \E[|E_i|] \le n^{1 + \frac{1}{2^{k+1}-1}} \left(\sum_{i=0}^{k} {2^{-i+2}}\right) = O(n^{1 + \frac{1}{2^{k+1}-1}}).
$$

\subsection{Stretch and Hopbound Analysis}

Let us first give an informal sketch of the analysis.
Let $a,b$ be vertices.  Choose an integer $r\ge 2$, 
and imagine dividing up the shortest $a$--$b$
path into $r^k$ intervals of length $\mu = \dist_G(a,b)/r^k$, 
where $\mu$ defines one ``unit'' of length.
Once $r$ and $\mu$ are fixed we prove that given \emph{any}
two vertices $u,v$ at distance at most $r^i \mu$,
there is \emph{either} 
an $h_i$-hop path from $u$ to $v$ with 
additive stretch $O(ir^{i-1})\cdot \mu$, 
\emph{or} there is an $h_i$-hop path from $u$ to a $V_{i+1}$-vertex with
length $(r^i + O(ir^{i-1}))\cdot \mu$.  Of course, when $i=k$ the set $V_{k+1}=\emptyset$ is empty, so we cannot be in the second case.
Since, by definition of $\mu$, $\dist_G(a,b) \le r^k\mu$,
there must be an $h_k$-hop path with additive stretch $O(kr^{k-1})\cdot\mu$.
In order for this stretch 
to be $\epsilon\dist_G(a,b)$ we must set $r=\Theta(k/\epsilon)$.

So, to recap, the integer parameter $r=\Theta(k/\epsilon)$ depends
on the desired stretch $\epsilon$, and $r$ determines the hopcount 
sequence $(h_i)$, which is defined inductively as follows.
\begin{align*}
    h_0 &= 1,\\
    h_i &= (r+1)h_{i-1} + r   & \mbox{for $i\in[1,k]$}.
\end{align*}

The parameter $\beta$ of the hopset is exactly $h_k$.  It is straightforward to show that $h_k < 2(r+1)^k$.
Once $r$ and $(h_i)$ are fixed,
Theorem~\ref{lemma:stretch-analysis} is proved by induction.

\begin{theorem}\label{lemma:stretch-analysis}
For any fixed real $\mu$ (the ``unit''), for all $i\in[0,k]$ and any 
pair $u, v\in V$ such that $\dist_G(u, v)\le r^i \mu$, 
at least one of the following statements holds.
\begin{enumerate}[(i)]
\item $\dist_{G\cup H}^{(h_i)}(u, v) \le \dist_G(u, v) + ((r+4)^i - r^i)\mu$,
\item There exists $u_{i+1}\in V_{i+1}$ such that $\dist_{G\cup H}^{(h_i)}(u, u_{i+1}) \le (r+4)^i\mu$.
\end{enumerate}
\end{theorem}

\begin{proof}
The proof is by induction on $i$.
In the base case $i=0$ and $h_0=1$. Let $u, v\in V$ with $\dist_G(u, v)\le r^0\mu = \mu$. If $(u, v)\in H$ then $\dist_{G\cup H}^{(1)}(u, v) = \dist_G(u, v)$ so (i) holds. Otherwise, $(u, v)\notin H$, meaning $v\notin \B(u)$.
If $u\in V_0 \setminus V_1$ then 
$\dist_{G\cup H}^{(1)}(u,p_1(u)) \le \dist_G(u,v) \le \mu$,
and if $u \in V_1$ then $p_1(u)=u$, so $\dist_{G\cup H}^{(1)}(u,p_1(u))=0$.
In either case, (ii) holds.

Now assume $i>0$.
Consider vertices $u, v\in V$ with $\dist_G(u, v)\le r^i \mu$
and let $P$ be a shortest $u$--$v$ path in $G$.
Then, as shown in Figure~\ref{fig:1}, we partition $P$ into at most $2r-1$ segments $\langle u_0=u, u_1\rangle$, $\langle u_1, u_2\rangle$, $\ldots$, $\langle u_{\ell-1}, u_\ell = v\rangle$ as follows.
Starting at $u_0=u$, we pick $u_1$ to be 
the \emph{farthest} vertex on $P$ such that $\dist_G(u_0,u_1) \le r^{i-1}\mu$,
and let $(u_1,u_2)$ be the next edge on the path.\footnote{Note that if the first edge has length more than $r^{i-1}\mu$, then $u_1=u_0$.}
Repeat the process until we reach $u_{\ell}=v$, oscillating between selecting segments that have length at most $r^{i-1}\mu$ and single edges.
\begin{itemize}\setlength{\itemsep}{0pt}
\item \emph{Multi-hop segment}: the shortest path from $u_s$ to $u_{s+1}$ satisfies $\dist_G(u_s, u_{s+1}) \le r^{i-1}\mu$.
\item \emph{Single-hop segment}: the segment is actually an edge $(u_s, u_{s+1})\in E$.
\end{itemize}

\begin{figure}
\centering
\includegraphics{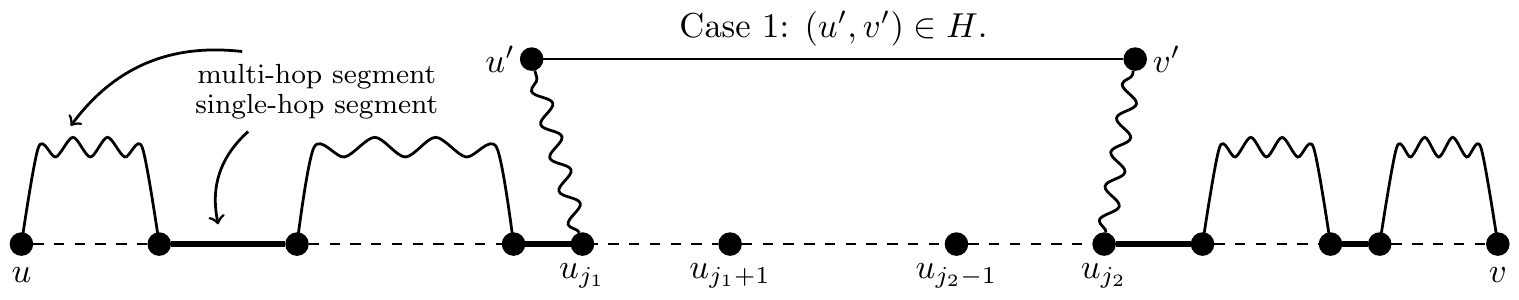}\\[0.5cm]
\includegraphics{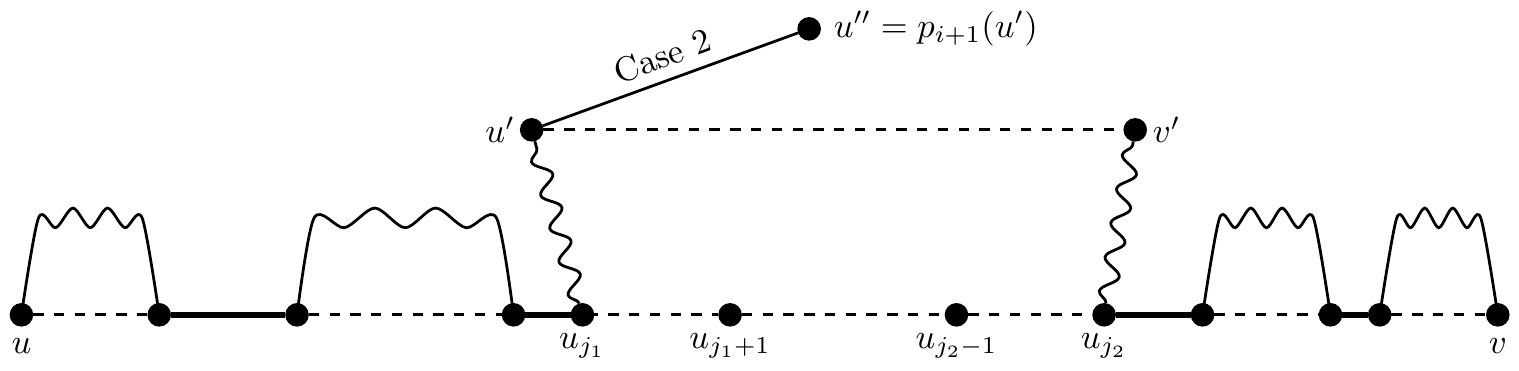}
\caption{The two cases depending on whether $(u', v')\in H$ or not. The first case leads to (i) and the second case leads to (ii) in the statement of Theorem~\ref{lemma:stretch-analysis}.}\label{fig:1}
\end{figure}

By the induction hypothesis, each multi-hop segment satisfies (i) or (ii) within $h_{i-1}$ hops.
Moreover, in each greedy iteration the sum of the lengths from picked multi-hop segment and immediately followed single-hop segment is strictly greater than $r^i\mu$ except the last one.
Therefore, by the pigeonhole principle, there are at most $r$ multi-hop segments on $P$ and at most $r-1$ single-hop segments on $P$.

If condition (i) holds for all multi-hop segments, then in at most $rh_{i-1}+r-1 \le h_i$ hops,
\begin{align*}
\dist_{G\cup H}^{(h_i)}(u, v) &\le \dist_G(u, v) + r((r+4)^{i-1} - r^{i-1})\mu\\
& \le \dist_G(u, v) + ((r+4)^i - r^i)\mu,
\end{align*}
and condition (i) holds for $P$.

Otherwise, condition (i) does not hold for at least one multi-hop segment. Consider the first multi-hop segment $\langle u_{j_1}, u_{j_1+1}\rangle$ and the last multi-hop segment $\langle u_{j_2-1}, u_{j_2}\rangle$ that do not satisfy condition (i). By condition (ii), there exist $u'$ and $v'\in V_i$ satisfying
\begin{align*}
    \dist_{G\cup H}^{(h_{i-1})}(u_{j_1}, u') &\le (r+4)^{i-1} \mu\\
    \dist_{G\cup H}^{(h_{i-1})}(u_{j_2}, v') & \le (r+4)^{i-1} \mu.
\end{align*}

Now we have two cases depending on whether $(u', v')\in H$ or not. If $(u', v')\in H$, then by the triangle inequality, we can get from $u_{j_1}$ to $u_{j_2}$ with $2h_{i-1}+1$ hops and additive stretch
\begin{align*}
    \dist_{G\cup H}^{(2h_{i-1}+1)}(u_{j_1}, u_{j_2}) - \dist_G(u_{j_1}, u_{j_2}) &\le \dist_{G\cup H}^{(h_{i-1})}(u_{j_1}, u') + \dist_H^{(1)}(u', v') + \dist_{G\cup H}^{(h_{i-1})}(v', u_{j_2}) - \dist_G(u_{j_1}, u_{j_2})\\
    &\le 2\dist_{G\cup H}^{(h_{i-1})}(u_{j_1}, u') + 2\dist_{G\cup H}^{(h_{i-1})}(v', u_{j_2})\\
    &\le 4(r+4)^{i-1}\mu. 
\end{align*}

We know there are a total of at most $r-1$ multi-hop segments satisfying condition (i). Hence, within at most $(r-1)h_{i-1} + r-1 + 2h_{i-1}+1\le h_i$ hops, we can get from $u$ to $v$ with additive stretch
\begin{align*}
    \dist_{G\cup H}^{(h_i)}(u, v) - \dist_G(u, v) &\le (r-1) ((r+4)^{i-1}-r^{i-1})\mu + \dist_{G\cup H}^{(2h_{i-1}+1)}(u_{j_1}, u_{j_2}) - \dist_G(u_{j_1}, u_{j_2})\\
&\le \left[(r-1)((r+4)^{i-1}-r^{i-1}) + 4(r+4)^{i-1}\right]\mu\\
&= \left[(r+3)(r+4)^{i-1} - r^i + r^{i-1}\right]\mu\\
&\le ((r+4)^i - r^i)\mu \tag{$r^{i-1}\le (r+4)^{i-1}$}
\end{align*}
and condition (i) holds for $P$ in this case.

On the other hand, suppose that $(u', v')\notin H$. Since both $u',v'\in V_i$ but $(u', v')\notin H$, we know that $u''=p_{i+1}(u')\in V_{i+1}$ must exist with $\dist_H^{(1)}(u', u'') \le \dist_G(u', v')$. Hence, we can get from $u_{j_1}$ to $u''$ via an $(h_{i-1}+1)$-hop path with length
\begin{align*}
\dist_{G\cup H}^{(h_{i-1}+1)}(u_{j_1}, u'') 
&\le\dist_{G\cup H}^{(h_{i-1})}(u_{j_1}, u') + \dist_{H}^{(1)}(u', u'') \\
&\le \dist_{G\cup H}^{(h_{i-1})}(u_{j_1}, u') + \dist_G(u', v')\\
&\le 2\dist_{G\cup H}^{(h_{i-1})}(u_{j_1}, u') + \dist_G(u_{j_1}, u_{j_2}) + \dist_{G\cup H}^{(h_{i-1})}(u_{j_2}, v')\\
&\le 3(r+4)^{i-1}\mu + \dist_G(u_{j_1}, u_{j_2}).
\end{align*}

Similar to the previous case, there are at most $r-1$ multi-hop segments appeared before $u_{j_1}$, and all of them are satisfying condition (i). Hence, the surplus
\begin{align*}
    \dist_{G\cup H}^{((r-1)h_{i-1}+r-1)}(u, u_{j_1}) 
    &\le \dist_G(u, u_{j_1}) + (r-1)((r+4)^{i-1} - r^{i-1})\mu.
\end{align*}

Therefore, in at most $(r-1)h_{i-1} + r-1 + h_{i-1}+1\le h_i$ hops,
\begin{align*}
\dist_{G\cup H}^{(h_i)}(u, u'') &\le \dist_{G\cup H}^{((r-1)h_{i-1}+r-1)}(u, u_{j_1}) + \dist_{G\cup H}^{(h_{i-1}+1)}(u_{j_1}, u'')\\
&\le \left[(r-1)((r+4)^{i-1}-r^{i-1}) + 3(r+4)^{i-1}\right]\mu + \dist_G(u, u_{j_2})\\
&\le \left[(r+2)(r+4)^{i-1} - r^i + r^{i-1}\right]\mu + \dist_G(u, u_{j_2})\\
&\le \left[(r+4)^i - r^i\right]\mu + \dist_G(u, u_{j_2})\tag{$r^{i-1}\le (r+4)^{i-1}$}\\
&\le (r+4)^i\mu \tag{$\dist_G(u, u_{j_2})\le \dist_G(u, v)\le r^i\mu$}
\end{align*}
\end{proof}

\begin{proof}[Proof of Theorem~\ref{thm:hopset}]
Fix $u, v\in V$ and $d = \dist_G(u, v)$. 
Define $\epsilon' = \ln (1+\epsilon)$. 
Notice that $1/\epsilon' = (1+o(1))(1/\epsilon)$. 
Set $r = \lceil 4k/\epsilon'\rceil = \Theta(k/\epsilon)$ and $\mu = d/r^k$. By Theorem~\ref{lemma:stretch-analysis}, since $V_{k+1}=\emptyset$, condition (i) must hold: within $h_k < 2(r+1)^k$ hops we have 
\begin{align*}
    d^{(h_k)}_{G\cup H}(u, v) &\le \dist_G(u, v) + ((r+4)^k - r^k) \mu\\
    &= d + \left(\frac{4k}{r} + \frac{4^2{k\choose 2}}{r^2} + \frac{4^3{k\choose 3}}{r^3} + \cdots \right) d\\
    &\le \left(1 + \epsilon' + \frac{\epsilon'^2}{2!} + \frac{\epsilon'^3}{3!} + \cdots \right) d \tag{since $4k/r \le \epsilon'$}\\
    &\le e^{\epsilon'} d = (1 + \epsilon) d.
\end{align*}
\end{proof}

Observe that if we set $k=\log\log n-O(1)$ the size becomes linear.

\begin{corollary}\label{cor:hopset}
Every $n$-vertex graph has 
an $O(n)$-size $(\beta,\epsilon)$-hopset with
$\beta = 2(\frac{(4+o(1))k}{\epsilon})^k$
and $k=\log\log n - O(1)$.
\end{corollary}

\section{Conclusion}

In this paper our goal 
was to demonstrate that hopset constructions need not be complex, 
and that optimal hopsets can be constructed with a \emph{simple} and elegant algorithm, namely a small modification to Thorup and Zwick's emulator construction~\cite{TZ06}.  From a purely quantitative perspective our hopsets also improve on the sparseness and/or hopbound of other constructions~\cite{Coh00,ElkinNeiman16,ElkinNeiman17}.  
As a happy byproduct of our construction, we also shave 
small factors off the best sublinear additive emulators~\cite{TZ06} and $(1+\epsilon,\beta)$-spanners~\cite{ABP17}.

We now have a good understanding of the tradeoffs available between $\beta$
and the hopset size when the stretch is fixed at $1+\epsilon$,
$\epsilon>0$ being a small real.  
However, when $\epsilon=0$ or $\epsilon$ is large, there are still gaps between the best upper and lower bounds.  For example, when $\epsilon=0$ a trivial 
hopset\footnote{Let $H$ be a clique on a set of $\sqrt{n}$ vertices chosen uniformly at random.} has size $O(n)$ with $\beta = O(\sqrt{n}\log n)$.
A construction of Hesse~\cite{Hesse03} (see also~\cite[\S 6]{ABP17}) implies that $\beta$ must be at least $n^\delta$ for some $\delta$, but it is open whether $O(n)$-size hopsets exist with $\beta \ll \sqrt{n}$.  At the other extreme, Thorup and Zwick's distance oracles imply that $O(\kappa n^{1+1/\kappa})$-size hopsets exist with $\beta=2$ and stretch $2\kappa-1$.  Is this tradeoff optimal?  Are there other tradeoffs available when $\beta$ is a fixed constant (say 3 or 4), 
independent of $\kappa$?

\paragraph{Acknowledgement.} Thanks to Richard Peng for help with the references for zero-stretch hopsets.  

\bibliography{main}

\end{document}